\documentclass[letterpaper,11pt]{article}
\usepackage{amssymb}
\usepackage{amsfonts}
\usepackage{amssymb,amsmath,amsthm}

\usepackage{subfig}
\usepackage{graphicx}

\usepackage[ruled]{algorithm}
\usepackage[noend]{algpseudocode}

\usepackage[letterpaper,nohead,margin={1in,1.2in}]{geometry}

\usepackage{setspace}
\usepackage{slashbox}
\usepackage{nonfloat}
\usepackage{subfig}

\usepackage{bm}
\newtheorem{theorem}{Theorem}[section]
\newtheorem{corollary}[theorem]{Corollary}
\newtheorem{lemma}[theorem]{Lemma}

\theoremstyle{definition}
\newtheorem{definition}[theorem]{Definition}

\theoremstyle{remark}

\numberwithin{equation}{section}

\begin{document}
\title{Partial Degree Bounded Edge Packing Problem}

\author{
Peng Zhang\thanks{Shanghai Key Laboratory of Trustworthy Computing, East
China Normal University, Shanghai, P.R.China. Email:
\texttt{arena.zp@gmail.com}\,.} }

\maketitle

\begin{abstract}
In \cite{zp12}, whether a target binary string $s$ can be represented from a boolean formula with operands chosen from a set of binary strings $W$ was studied. In this paper, we first examine selecting a maximum subset $X$ from $W$, so that for any string $t$ in $X$, $t$ is not representable by $X \setminus \{t\}$. We rephrase this problem as graph, and surprisingly find it give rise to a broad model of edge packing problem, which itself falls into the model of forbidden subgraph problem. Specifically, given a graph $G(V,E)$ and a constant $c$, the problem asks to choose as many as edges to form a subgraph $G'$. So that in $G'$, for each edge, at least one of its endpoints has degree no more than $c$. We call such $G'$ partial $c$ degree bounded. When $c=1$, it turns out to be the complement of dominating set. We present several results about hardness, approximation for the general graph and efficient exact algorithm on trees. This edge packing problem model also has a direct interpretation in resource allocation. There are $n$ types of resources and $m$ jobs. Each job needs two types of resources. A job can be accomplished if either one of its necessary resources is shared by no more than $c$ other jobs. The problem then asks to finish as many jobs as possible. We believe this partial degree bounded graph problem merits more attention.
\end{abstract}

\section{Introduction}
An elementary problem of set operations is stated as follows. Given a collection of subsets of the universe, what new subsets can be generated if \emph{union} and  \emph{intersection} are allowed. By denoting the subset as an indicator vector, we reformulate it like this. Given two binary strings with the same length, namely $s_1, s_2$, let $s_1 \wedge
s_2$ (resp. $s_1 \vee s_2$) be the binary string produced by bitwise AND $\wedge$
(resp. OR $\vee$) of $s_1$ and $s_2$. Given a set of $m$ bits long binary strings,
namely, $W=\{s_1,s_2, \cdots, s_n\}$, $s_i \in \{0,1\}^m$, if there is a
formula $\phi$ which calculates $s$, with operators in $\{ \wedge, \vee \}$ and
operands in some subset of $W$, then we say the target string $s$ is
representable by (or expressible from) $W$ via formula $\phi$, or simply $s$ is representable.

A natural variant of this problem is finding a maximum subset, in which each string is not representable by the others. We call this variant Maximum Expressive Independent Subset (MEI) problem and examine the restricted case on strings with exactly two ones. Surprisingly, this is equivalent as maximum edge packing under partial degree bounded by 2.

This paper is structured as follows. We study the hardness of edge packing bounded by 1, by 2 and by a constant less than $\Delta(G)$ on graph in section \ref{sec:MEIS}. Then we study the general edge packing on trees in section \ref{sec:PcBonTree}. In section \ref{sec:approxAlgs}, approximation algorithms for bounded 1 and bounded 2 edge packing are built. Some conclusions are given in section \ref{sec:conclusion}. Because the problem only concerns edges selecting, we assume the graph we deal with is free of isolated vertex.

\subsection{Related work}
The decision problem of edge packing bounded by 1 turns out to be a parametric dual of the well known dominating set(DS). The parametric dual means that for graph $G(V,E)$, a $k$ sized dominating set implies a $|V|-k$ sized edge packing, and vice versa. The parametric dual of DS was studied in \cite{Nieminen74boundsForDomination}, in which the edges packed are called pendant edges. Further, the dual was well studied in the framework of parameterized complexity by Frank Dehne, etc in \cite{Dehne06nonblocker:parameterized}. They coined the dual as NONBLOCKER problem and showed a linear kernel of $5/3\bm\cdot k_d +3$, where $k_d$ is the solution size.

\section{Maximum Expressible Independent Subset}\label{sec:MEIS}
At first, we introduce some notations used in \cite{zp12}. Let $x$ denote any binary string, $b_{i}^{x}$ denote the $i^{\text{th}}$ bit of
$x$. So, $x=b_{1}^{x}b_{2}^{x} \cdots b_{m}^{x}$. Also, we define a function
$\mathsf{Zero}: \mathsf{Zero}(x)=\{i|b_{i}^{x}=0\}$, from a binary string to a
set of natural numbers which denotes the indices of bits with value 0 in the
binary string. Similarly, $\mathsf{One}(x)$ denotes the indices of 1 valued
bits of $x$. Also, $\mathbf{0}$ (resp. $\mathbf{1}$) denotes a binary string
with no 1 (resp. 0) valued bits. Let $N_i$ denote the set of strings whose $i^{th}$ bit is 1, i.e,
$N_i=\{y|b_i^y=1,y \in W\}$. In addition, $T_i$ denotes the set of binary strings in $W$ whose
$i^{\text{th}}$ bit value is 0, i.e., ${T_i} = \{x \in W|b^x_i=0\}$. Let ${t_i}
= \bigvee_{x \in {T_i}} x$.

\begin{definition}[Expressible Independent Set (EI)]
A set $X$ of binary strings is expressible independent if and only if for each binary string
$x \in X$, $x$ is not expressible from $X\setminus \{x\}$.
\end{definition}

Then the Maximum Expressible Independent Subset (MEIS) problem is defined as follows. Given a set $W$ of binary strings, MEIS asks to find a maximum expressible independent subset of $W$. The decision version with parameter $k$ is denoted as MEIS$(W,k)$.

\subsection{MEIS on $2$-regular set}
We first pay attention to a restricted case of MEIS, when each binary string has the same number of bits valued 1.
And we refer to the following theorem \ref{lem:basic} from \cite{zp12}.

\begin{theorem}\label{lem:basic}
Given $(W,s)$ where $s \neq \mathbf{1}$, then $s$ is expressible from $W$ if and only if
$\forall i \in \mathsf{Zero}(s)$, $\mathsf{One}(s) \subseteq
       \mathsf{One}(t_i)$.
\end{theorem}

\begin{definition}[c-regular set]
A binary string is $c$-regular if and only if it contains exactly $c$ one bits. A set of binary strings is $c$-regular if and only if each element is $c$-regular.
\end{definition}

\begin{lemma} \label{lem:2regular}
Given a \emph{$2$-regular} set $W$ and a \emph{$2$-regular} string $x \not\in W$, $\mathsf{One}(x)=\{i,j\}$, then $x$ is expressible from $W$ if and
only if $|N_i|\ge 2$ and $|N_j|\ge 2$.
\end{lemma}
\begin{proof}
\emph{Sufficiency}: We prove its contrapositive. By symmetry, suppose
that  $|N_i|\le 1$. If $|N_i|=0$, then $ i \not\in
\mathsf{One}(t_l), l\in \mathsf{Zero}(x)$. If $N_i=\{y\}$, assume
that $\mathsf{One}(y)=\{l,i\}$, then $ i \not\in \mathsf{One}(t_l)$. In both cases,
$\mathsf{One}(x) \nsubseteq \mathsf{One}(t_l)$, thus $x$ is not expressible from $W$ according to Lemma \ref{lem:basic}. 

\emph{Necessity}: If $|N_i|\ge 2$ and $|N_j| \ge 2$, we assume
$\{a,b\}\subseteq N_i$ and $\{c,d\}\subseteq N_j$. It is easy to
check that $(a\land b)\lor(c \land d)=x$.
\end{proof}

\begin{definition}[Partial Degree Bounded Graph]
An undirected graph $G(V,E)$ is partial $c$ bounded (P$c$B) if and only if $\forall_{e(u,v) \in E}{ (d_u \le c \bigvee d_v \le c)}$. $d_u$ is the degree of $u$.
\end{definition}

Given a graph $G$, the Maximum Partial $c$ Degree Bounded Graph problem asks to find a P$c$B subgraph $G'$ of $G$ with maximum edges. The decision version with parameter $k$ is denoted as P$c$B$(G,k)$. In the setting of resource allocation, each vertex stands for a resource, each edge stands for a job. And an optimum P$c$B subgraph maps to an optimum resource allocation.

Now, we will rephrase MEIS$(W,k)$ on $2$-regular set as P2B$(G,k)$. Let $W \subseteq
\{0,1\}^m$, we construct the corresponding graph $G(V,E)$
as follows. Vertex $v_i \in V$ corresponds to the $i^{th}$
bit of string. Each edge $(v_i,v_j)\in E$ corresponds to a
string $x \in W$ whose $\mathsf{One}(x)=\{i,j\}$. According to Lemma
\ref{lem:2regular}, MEIS$(W,k)$ has a solution if and only if
$P2B(G,k)$ has a solution. Just select the corresponding edges in $G$, and select the corresponding strings in $W$ vice versa. The reduction can be done in the reverse way. So it is just a rephrasing.

\begin{lemma}\label{lem:PartialOne}
$P1B(G,k)$ is NP-complete.
\end{lemma}
\begin{proof}
Given a graph $G(V,E)$, P1B$(G,k)$  is in NP
trivially because we can check in $O(|E|)$ time that whether the
given subgraph $G'$ is partial 1 bounded. We prove its NP-completeness
by showing that, there is a a $k$ sized partial 1 bounded subgraph $G'$ if and only if
there is a $n-k$ sized dominating set $D$ of $G$, $n=|V|$. Note
that, any partial 1 bounded graph is a set of node-disjoint stars.

\emph{Necessity}: If $D=\{v_1,\cdots,v_k\}$ is a dominating set,
then we can construct a $k$ node-disjoint stars as following, which is a partition of $G$. Let
$P_i(V_i,E_i)$ denote the $i^{th}$ star being constructed. For each
vertex $u \in V\setminus D$, if $u$ is dominated by $v_i$, add $u$
into $V_i$ and $(v_i,u)$ into $E_i$. To make the stars
node-disjoint, when $u$ is dominated by more than one vertices,
break the ties arbitrarily. Note that, $E_i$ may be empty, that is,
the $P_i$ only contains an isolated vertex. So $\sum_{i\le k}{|E_i|}=\sum_{i\le k}{|V_i|}-k=n-k$. Thus $\bigcup_{v_i \in D}{P_i}$ is a $n-k$ sized partial 1 bounded graph.

\emph{Sufficiency}: If there is a $G'$ with $|E_{G'}|=n-k$. Suppose that
$G'$ contains $n_0$ stars without leaf (i.e., isolated vertices) and
$n_1$ stars with at least one leaf. It is easy to see, $n-k =
(n-n_0)-n_1$. Thus $n_0+n_1=k$, so we just select the isolated
vertices and the internal node of the $n_1$ stars. They make up a $k$ sized dominating set.
\end{proof}

\begin{lemma}\label{lem:P2BNPC}
$P2B(G,k)$ is NP-complete, so is MEIS$(W,k)$ on 2-regular set.
\end{lemma}
\begin{proof}
This problem is trivial to be in NP. We show its NP-completeness via
a reduction from $P1B(G,k)$. Given a $P1B(G,k)$ instance
$G(V,E)$, $n=|V|$, we construct a $P2B(G',n+k)$ instance $G'(V',E')$ as follows. Adding a distinguished vertex $u$ into $V$, i.e., $V'=V\cup \{u\}$ and $E'=E\cup\{(u,v)|v \in V\}$.

\emph{Necessity}: If $M$ is a $k$ sized partial 1 bounded subgraph in $G$, then adding the $n$ additional edges, i.e.,
$E'\setminus E$ into $M$ will produce a $n+k$ sized partial 2 bounded subgraph $M'$.

\emph{Sufficiency}: Let $M'$ be a $n+k$ sized partial 2 bounded subgraph in $G'$, and let $d_{v}^{M'}$ be the degree of $v$ in $M'$. We prove it case by case. \textbf{Case 1:} When $(E'\setminus E) \subseteq E_{M'}$, then deleting all the $n$ additional edges will make each node's degree in $M'$ decrease 1, thus the remaining subgraph is a $k$ sized partial 1 bounded subgraph. \textbf{Case 2:} When there exists an edge $(u,v_i)\not\in E_{M'}$, if $d_{v_i}^{M'} < 2$, then we could just replace an arbitrarily edge $(v_j, v_k) \in E_{M'}$ by $(u,v_i)$. If $d_{v_i}^{M'} \ge 2$, let $(v_i, v_j)\in E_{M'}$, then we could replace $(v_i, v_j)$ by $(u,v_i)$. Repeating this swap, we eventually arrive at a $n+k$ sized partial 2 bounded subgraph $M'$ which contains all the $n$ addition edges, that is Case 1. 
\end{proof}

\begin{theorem}
P$c$B$(G,k)$ is NP-complete.
\end{theorem}
\begin{proof}
We can easily generalize the proof technique of \label{lem:P2BNPC} to any parameter $c$. Given a P1B$(G,k)$ instance $G(V,E)$, $n=|V|$, we add $c-1$ additional vertices and $(c-1)|V|$ edges connecting each of the vertex in $V$ to each additional vertex, resulting a graph $G'$. Then P1B$(G,k)$ is a YES instance if and only if P$c$B$(G',k+(c-1)|V|)$ is a YES instance. Note that, it is trivial to select all the edges when $c \ge \Delta(G)$, where $\Delta(G)$ stands for the maximum degree of $G$.
\end{proof}

It is clear that P$c$B$(G,k)$ is a case of forbidden subgraph of $G$, which asks to find a maximum subgraph $G'$ of $G$ so that $G'$ does contain a subgraph which is isomorphic with the forbidden $H$. Here $H$ is a tree with 2 internal nodes with degree $c+1$ and $2c$ leaf nodes. Each internal node has incident edges to $c$ leaf nodes and the other internal node.

\section{Maximum partial $c$ bounded subgraph problem on tree}\label{sec:PcBonTree}
Due to the NP-hardness of P2B on the general graph, we would first consider it on some restricted structures, such as tree. In the
scenario of Maximum Expressible Independent Subset, this corresponds to restricted instances where for any subset $A \subseteq W$, $|\bigcup_{x\in A}{\mathsf{One}(x)}| > |A|$.

\begin{lemma}
$PcB(G,k)$ is solvable in $O(n^2)$ for any parameter $c$ via a dynamic programming, where $n=|V|$.
\end{lemma}
\begin{proof}
The sketch of the algorithm is bottom-up for the tree as a whole, and left to right knapsack like dynamic programming for selecting a vertex's children. Let $T(V,E)$ denote the tree, and
$v_1,\cdots,v_n$ be a breadth first ordering of vertices of $V$.
Further, let $d_u'$ be the number of $u$'s children and let $T(u, i)$
be the subtree induced by $u$, $u$'s first $i$ children and
all their descendants. Thus, $T(u, d_u')$ is the subtree rooted at
$u$ and $T(u,0)$ contains $u$ alone. Let $f(u,i,q)$ denote the maximum P$c$B subtree in $T(u,i)$ under the condition that $u$ has $q$ neighbors in $T(u,i)$. Let $g_1(u)$ denote the maximum P$c$B subtree in $T(u,d_u')$ under the condition that $u$ has less than $c$ neighbors, and $g_2(u)$ for exactly $c$ neighbors and $g_3(u)$ for more than $c$ neighbors respectively. So only $g_1(u)$ and $g_3(u)$ can be extended to have an edge connecting to $u$'s parent when we are working upward. For simplicity, we abuse $f(u,i,q)$ and $g_i(u)$ to denote their edge cardinalities. Let $\mathsf{MAX}\{\cdots , a_i, \cdots \} = \arg max_i \{a_i\}$.

\begin{algorithm}
\caption{Solving Partial $c$ Bounded Subgraph on Trees}
\begin{algorithmic}[1]
\ForAll{$u$, $u$ is a leaf}

\State {$f(u,0,0)=0$}

\EndFor

\ForAll{$u$, all subtrees rooted at $u$'s children have been calculated}

    \ForAll{ $i$ from $1$ to $d_u'$}

        \State {Let $v$ be the $i^{th}$ child of $u$ counting from left to right}

        \ForAll{ $q$ from $1$ to $i$ }

            \State {$f(u,i,q)= f(u,i-1,q)+ \mathsf{MAX}\{ g_1(v), g_2(v), g_3(v)$ \} }

            \If { $q \le c$ }

            \State { $f(u,i,q) = \mathsf{MAX} \{ f(u,i,q), f(u,i-1,q-1) + 1 +\mathsf{MAX}\{ g_1(v),g_3(v) \} \}$ }

            \Else

            \State { $f(u,i,q) = \mathsf{MAX} \{ f(u,i,q), f(u,i-1,q-1) + 1 + g_1(v) \}$ }

            \EndIf

        \EndFor

        \State { update $g_1(u)$, $g_2(u)$ and $g_3(u)$ }

    \EndFor

\EndFor

\end{algorithmic}
\end{algorithm}

The algorithm above is correct because it \emph{enumerates} every possible
edge selection by a knapsack like way. Lines 1-2 take $O(|V|)$. It is important to note that lines
3-12 take only $\Sigma (d_i')^2 \le (\Sigma d_i')^2 \le (2|V|)^2$. So the running
time of the algorithm is $O(|V|^2)$.
\end{proof}

\section{Approximation algorithms for P1B and P2B}\label{sec:approxAlgs}
In this section, we are going to present two approximation algorithms. The first one for partial 1 bounded subgraph runs in $O(|E|)$ with approximation ratio 2, and the one for partial 2 bounded runs in $O(|V|)$ with ratio $\frac{32}{11}$ in expectation. In analyzing both algorithms, we only use upper bounds of the optimum solutions, without exploring deep relationships between the optimum and the solution our algorithm returned.

\subsection{A 2 ratio approximation algorithm for partial 1 bounded subgraph}
Given a graph $G(V,E)$, we first greedily calculate a dominating set with no more than $|V|/2$ vertices and then construct a partial 1 bounded subgraph $M$ with no less than $|V|/2$ edges. Because the maximum P1B subgraph of $G$ has less than $|V|$ edges, $M$ is a 2 ratio approximation solution. The process is shown in algorithm \ref{alg:approxPartial1}.

\begin{algorithm}
\caption{Approximation Algorithm for maximum partial 1 bounded subgraph of $G(V,E)$}\label{alg:approxPartial1}
\begin{algorithmic}[1]
\State{$D\subseteq V$, $A\subseteq V$, initiate $A=D=\emptyset$ }

\ForAll{$u \in V$}
\If{$u$ is not dominated by any vertex $v\in D$}
\State{$D = D\bigcup \{u\}$}
\EndIf
\EndFor

\If{$|D| > |V|/2$}
\State{$D=V\setminus D$}
\EndIf

\ForAll{ $u \in D$}
\ForAll{edge $(u,v) \in E$}
\If{ $v \in D$ or $v \in A $}
\State{delete $(u,v)$}
\EndIf
\If{ $v \not\in D$ and $v \not\in A $ }
\State{ $A=A\bigcup \{v\}$ }
\EndIf
\EndFor
\EndFor

\State{$G$ is a P1B graph}
\end{algorithmic}
\end{algorithm}

Lines 1-4 in algorithm \ref{alg:approxPartial1} obtains a minimal dominating set (DS) $D$ in $O(|E|)$. Then $V\setminus D$ is also a minimal DS and lines 5-6 obtains a minimal DS with no more than half vertices. Lines 7-12 obtains a P1B subgraph of $G$, this is proved in lemma \ref{lem:PartialOne}.

\subsection{ A 32/11 ratio approximation algorithm for partial 2 bounded subgraph}
We first present an upper bound of the optimum value and then give a randomized algorithm with expectation larger than $\frac{11}{32}$ times the upper bound. Eventually we show the process of derandomization. Let $N(u)=\{v | (u,v) \in E\}$ denote the neighbors of $u$ in graph $G$. Also, if $A$ is a set of vertices, then let $N(A)=\bigcup_{u \in A}{N(u)}$. In the sequel, $n=|V|$.

\begin{lemma}\label{lem:d_uOrd_v>c}
If $G(V,E)$ is a maximum partial $c$ bounded graph on $n$ vertices, then \\ $\forall_{(u,v)\in E}{(d_u \ge c \bigwedge d_v \ge c)}$.
\end{lemma}
\begin{proof}
If there is an edge $(u,v)$ dissatisfies $(d_u \ge c \bigwedge d_v \ge c)$, i.e., $(d_u < c \bigvee d_v <c)$, we assume $d_u < c$. Let $X=V\setminus N(u,v)$, then $X\neq \emptyset$ because $u$ and $v$ can have at most $2c-2$ neighbors. Otherwise, we can construct a graph with more edges. We do case by case proof as follows. \textbf{Case 1:} If $\exists_{x \in X}{d_x > c}$, then we apply an edge addition as $E=E \bigcup \{(u,x)\}$. This edge addition preserves $G$'s property as a P$c$B, we call it \emph{valid}. \textbf{Case 2:} If $\forall_{x \in X}{d_x \le c}$, then $d_x \le c$, we apply the same edge addition as in Case 1. Both edge additions contradict the fact that $G$ is a maximum P$c$B. So the lemma holds.
\end{proof}

According to lemma \ref{lem:d_uOrd_v>c}, $\forall_{u\in V}{(d_u \ge c)}$. By definition of P$c$B, $\forall_{(u,v)\in E}{(d_u \le c \bigwedge d_v \le c )}$, then at least one endpoint has degree $c$. So the following corollary is correct.
\begin{corollary}\label{cor:d_uOrd_v=c}
If $G(V,E)$ is a maximum partial $c$ bounded graph on $n$ vertices, then \\ $\forall_{(u,v)\in E}{(d_u = c \bigvee d_v = c)}$.
\end{corollary}

\begin{lemma}\label{lem:d_u=cANDd_v>c}
If $G(V,E)$ is a maximum partial $c$ bounded graph on $n$ vertices, then \\ $\forall_{(u,v)\in E}{(d_u = c \bigwedge d_v > c)}$.
\end{lemma}
\begin{proof}
If there is an edge $(u,v)$ dissatisfies $(d_u = c \bigwedge d_v > c)$, then according to lemma \ref{lem:d_uOrd_v>c} and corrolary \ref{cor:d_uOrd_v=c}, the only possibility is $(d_u=c \bigwedge d_v =c)$. Let $X=V\setminus N(u,v)$, then $X\neq \emptyset$. We do case by case proof as follows. \textbf{Case 1:} If $\exists_{x \in X}{d_x > c}$, then we can apply an \emph{valid edge augmentation} as $E = (E\setminus \{(u,v)\}) \bigcup \{(u,x),(v,x)\}$.  \textbf{Case 2:} If $\forall_{x\in X}{d_x=c}$, then we can choose an $x$ arbitrarily and apply the same \emph{valid edge augmentation} as in Case 1. However, both edge augmentations contradict that $G$ is a maximum P$c$B. So the lemma holds.
\end{proof}

\begin{theorem}\label{thm:PcBupperbound}
For any partial $c$ bounded graph $G(V,E)$, $|E| \le c\bm\cdot(|V|-c)$.
\end{theorem}
\begin{proof}
Let $G(V,E)$ be a maximum P$c$B graph on $n$ vertices. With the help of lemma \ref{lem:d_u=cANDd_v>c}, we can calculate the number of vertices having degree $c$. Let $y$ denote this number. Suppose $y > n-c$, then there are less than $c$ vertices with degree more than $c$. Thus for any $d_u =c$, $u$ can only have less than $c$ neighbors, which contradicts $d_u=c$. So $y \le n-c$, and $|E| \le c\bm\cdot y \le c\bm\cdot(|V|-c)$. When $y=n-c$, we can easily construct a P$c$B graph with $c\bm\cdot(|V|-c)$ edges. So the theorem holds.
\end{proof}

We justify an assumption that $\forall_{u\in G}(d_u^G>2)$ as follows. Let $E'=\{(u,v) | d_u^G \le 2\}$ and $M(V_M, E_M)$ be a partial 2 bounded subgraph. If $E' \subseteq E_M$, our assumption holds because we only need to consider the graph with minimum degree larger than 2. Otherwise, we repeatedly do the following swap in and out operations till $E' \subseteq E_M$. Let $(u,v_i) \in E'\setminus E_M$ and $(v_i, v_j)\in E_M\setminus E'$, then we could $E_M=\left(E_M\setminus \{(v_i,v_j)\} \right) \cup \{(u,v_i)\} $, i.e., swap $(v_i,v_j)$ out of $M$ and swap $(u,v_i)$ in $M$. The replaced $M$ is also a P$2$B subgraph.

\begin{algorithm}
\caption{Randomized Algorithm for maximum partial 2 bounded subgraph}\label{alg:randBipartite}
\begin{algorithmic}[1]
\State{$B(V_B,E_B)$ is a bipartite graph, $V_B=L\bigcup R$, initiate $L=R=\emptyset$ }
\ForAll{$u \in V$}
\State{add $u$ into $L$ or $R$ with equal probability $1/2$}
\EndFor
\ForAll{$e(u,v) \in E$ with $u \in L$}
\If{$u$ has no more than 2 edges in $E_B$}
\State{add $e$ into $E_B$}
\EndIf
\EndFor
\end{algorithmic}
\end{algorithm}

It is clear that $B(V_B,E_B)$ is a partial 2 bounded subgraph of $G$. Now we will analyze the size of $E_B$. Let $f(u)$ be the degree of $u$, $u \in L$, so $|E_M|=\Sigma_{u\in L}{f(u)}$. And let $d_u$ be the degree of $u$ in $G$, the expectation of
$f(u)$ is $\mathsf{E}[f(u)]=\frac{1}{2}\left( 1\cdot\frac{d_u}{2^{d_u}} + 2\cdot\left(
1 - \frac{1 + d_u}{2^{d_u}} \right) \right) = 1 - \frac{2 +
d_u}{2^{d_u+1}} \ge \frac{11}{16}$. Using the linear addition property,
$\mathsf{E}[E_M]=\Sigma_{u\in V}{\mathsf{E}[f(u)]} \ge \frac{11}{16}n > \frac{11}{32}OPT$,
where \emph{OPT} denotes the optimum value. According to theorem \ref{thm:PcBupperbound} conditioned on $c=2$, $OPT < 2n$ and the last inequality holds.

Lines 1-3 take up $O(|V|)$ time and lines 4-6 take up $O(|E|)$ time, so algorithm \ref{alg:randBipartite} takes $O(|E|)$ time. Because it is not direct to show whether the variables $\{f(u)| u\in V\}$ are independent or with small dependency, so we are not sure whether $|E_M|$ is sharply concentrated around its expectation in $O(|E|)$. But we can de-randomize algorithm \ref{alg:randBipartite}, using conditional expectation to decide  whether the next vertex should be put in $L$ or not. And the cost for deterministic algorithm is $O(|E|^2)$.

\section{Conclusion}\label{sec:conclusion}
This paper presents a new model of edge packing problem with partial degree bounded constraint and several results on it. The author is still trying to study more deep results in the following respects.

\textbf{P$c$B in a parameterized view}\\
When $c=1$, P$c$B is fixed parameter tractable (FPT) with respect to its solution size. Does this hold for general $c$? When $c$ is a constant, i.e., $c=o(|V|)$, it is easy to show P$c$B is in $W[1]$ defined in \cite{Downey98parameterizedcomplexity}. For example, when $c=2$, for each forbidden subgraph of P2B, we create a antimonotone clause $(\overline {e_1} \bigvee \overline {e_2} \bigvee \cdots \overline {e_5})$ where each literal $\overline {e_1}$ corresponds to an edge in the subgraph. Thus there is a P$c$B subgraph with $k$ edges if and only if the weighted 5-CNF satisfiability has a valid truth assignment with $k$ variables being set true. Because weighted 5-CNF satisfiability is $W[1]$-complete, so P$c$B is in $W[1]$.


According to theorem \ref{thm:PcBupperbound}, the solution may be close to $c$ times $n$ which renders the solution size not a good parameter. For example, when $c=2$, let $k$ be the parameter. Suppose $k < \Delta$, let $M$ be a maximum matching of $G$. Thus $k > |M| > \frac{n}{\Delta} > \frac{n}{k}$. So $k> \sqrt{n}$ and $\sqrt{n}$ is certainly not a good parameter.

Section \ref{sec:PcBonTree} shows that P$c$B is in P on trees, whether P$c$B could be efficiently (though not in P) solved on tree-like graph? Tree decomposition in \cite{Robertson198449} is a measure for this. Courcelle's theorem in \cite{Courcelle90} asserts that if a graph problem could be described in monadic second order (MSO) logic, then it could be solved in linear time with respect to its treewidth. Luckily, P$c$B is in MSO and thus establishes its FPT with respect to the treewidth. The author is trying to design a P$c$B specific algorithm with improved efficiency.

\textbf{P$c$B in a approximation view}\\
In section \ref{sec:approxAlgs}, we only show algorithms which upper bounds optimum roughly. We might elaborate the analysis by correlate the optimum with the solution returned by our algorithm. Also, both algorithms can not be extended when $c$ increases. Constant ratio approximation algorithms for the general $c$ or inapproximability results which exclude them would be really interesting.

Special thanks to Jukka Suomela and Chandra Chekuri for their valuable advice.

\bibliographystyle{abbrv}
\bibliography{PDBG}
\end{document}